\documentclass[reprint,aps,pra,groupedaddress,showpacs,longbibliography%
]{revtex4-1}

\usepackage{amsthm}
\usepackage{amsmath}
\usepackage{amssymb}
\usepackage{empheq}
\usepackage{xspace}

\DeclareMathOperator{\Span}{span}
\DeclareMathOperator{\Tr}{Tr}

\newcommand{\IIS}{IrIS\xspace}
\newcommand{\POVM}{POVM\xspace}

\newtheorem{lemma}{Lemma}
\newtheorem{theorem}{Theorem}
\newtheorem{corollary}{Corollary}
\newtheorem{remark}{Remark}
\newtheorem{example}{Example}

\begin{document}

\title{Structural Characterization And Condition For Measurement Statistics
Preservation Of A Unital Quantum Operation}

\author{Kai-Yan Lee} \thanks{Present address: Department of Astronomy and
 Oskar Klein Centre for Cosmoparticle Physics, Stockholm University,
 Albanova, SE-10691 Stockholm, Sweden} \email{lee.kai\_yan@astro.su.se}
\author{Chi-Hang Fred Fung} \email{chffung@hku.hk}
\author{H.\ F. Chau}\email[Corresponding author, ]{hfchau@hku.hk}
 \affiliation{Department of Physics and Center of Computational and Theoretical Physics\\
 University of Hong Kong, Pokfulam Road, Hong Kong}
 
\date{\today}

\begin{abstract}
We investigate the necessary and sufficient condition for a convex cone of
positive semidefinite operators to be fixed by a unital quantum operation
$\phi$ acting on finite-dimensional quantum states.
By reducing this problem to the problem of simultaneous diagonalization of the
Kraus operators associated with $\phi$, we can completely characterize the kind
of quantum states that are fixed by $\phi$.
Our work has several applications.
It gives a simple proof of the structural characterization of a unital quantum
operation that acts on finite-dimensional quantum states --- a result not
explicitly mentioned in earlier studies.
It also provides a necessary and sufficient condition for what kind of
measurement statistics is preserved by a unital quantum operation.
Finally, our result clarifies and extends the work of St{\o}rmer by giving a
proof of a reduction theorem on the unassisted and entanglement-assisted
classical capacities, coherent information, and minimal output Renyi entropy of
a unital channel acting on finite-dimensional quantum state.
\end{abstract}


\maketitle

\section{Introduction}\label{Sec:intro}

A quantum channel can be modeled by a quantum operation $\phi$ on a separable
Hilbert space ${\mathcal H}$.
Mathematically, $\phi\colon {\mathcal B}({\mathcal H}) \to
{\mathcal B}({\mathcal H})$ is a trace-preserving completely positive map on
the set of all bounded operators of ${\mathcal H}$.
Characterizing quantum operations and studying their properties are two
important areas of research in quantum information science.
Nevertheless, many apparently simple and basic questions regarding a quantum
operation are difficult to answer.
One example is to find all quantum states $\rho$ (in other words, trace one
positive self-adjoint operators) that are fixed points of a given quantum
operation in the sense that $\phi(\rho) = \rho$.
This question is still open.
Recently, some progress has been made on attacking this question.
Kribs~\cite{Kribs03}, who used simple functional analysis argument, and Arias
\emph{et al.}~\cite{AriasJMP02}, who applied the generalized L\"{u}ders
theorem, independently discovered a useful necessary and sufficient condition
for a class of quantum channels known as unital quantum operations to fix a
quantum state provided that the dimension of the Hilbert space ${\mathcal H}$
is finite.
Studies on the generalization and limitations of Arias \emph{et al.}'s result
along the line of generalized L\"{u}ders theorem have also been
reported.~\cite{DuJMP08,LiuJMP09,LiuJPA10,LiJMAA11,LongJPA11,LiJMP11,
PrunaruJPA11,LongJPTP11}

Restricting the study to unital rather than general quantum channels is a
sensible tactic to make progress.
This is because many physical processes in actual experiments such as
depolarization and dephasing can be modeled by unital channels so that it is
worthwhile to study these channels.
Besides, the mathematical tools and results to deal with unital quantum
operations are reasonably well-developed.
Therefore, it is not surprising that several quantum information science
problems involving unital quantum operations have been solved.
For example, by means of finite-dimensional $C^*$-algebra, St{\o}rmer found
that the evaluation of the (unassisted) classical capacity $C_{1,\infty}$ of a
unital quantum channel acting on finite-dimensional quantum states can be
reduced to calculating the same channel capacity for the case in which the
image of a non-scalar projection is never a projection.~\cite{St07}
In another study, Blume-Kohout \emph{et al.} used unital quantum channel as an
auxiliary tool and results from matrix algebra to characterize the geometric
structure of noiseless subsystems, decoherence-free subspaces, pointer bases
and quantum error-correcting codes acting on finite-dimensional quantum states.
The geometric structure they have identified is an isometry to fixed points of
certain unital quantum operations.~\cite{BK08,BK10}
They also extended their results by showing several interesting properties of
the set of fixed points of an arbitrary finite-dimensional quantum
channel.~\cite{BK10}
Using similar techniques, Rosmanis studied the properties of the fixed space of
a positive (but not necessarily completely positive) trace-preserving
map.~\cite{Rosmanis11}
Recently, Mendl and Wolf discovered several equivalent definitions for a unital
quantum operation and used them to study the relation between unital channels
and quantum error corrections.
Their results shine new light on the asymptotic Birkhoff conjecture and the
separation of the set of mixtures of unitary channels.~\cite{MendlWolf}
Lately, Zhang and Wu showed an easily checkable necessary and sufficient
condition for a finite-dimensional quantum state whose von~Neumann entropy is
preserved by a given unital channel.
Interestingly, their work is closely related to finding fixed points of the
composition of a unital quantum operation $\phi$ and its adjoint quantum
operation.~\cite{LW11}

In turns out that the works of Kribs~\cite{Kribs03}, St{\o}rmer~\cite{St07} and
Blume-Kohout \emph{et al.}~\cite{BK08,BK10} mentioned in the above paragraph
are closely related to the following variation of the fixed point problem in
which we called the fixed convex cone of positive semidefinite operator
problem.
Let ${\mathcal B}({\mathcal H})$ be the set of all linear operators of a
finite-dimensional Hilbert space ${\mathcal H}$.
(All Hilbert spaces considered in this paper are of finite dimensions.)
And let $\phi \colon {\mathcal B}({\mathcal H}) \to {\mathcal B}({\mathcal H})$
be a unital quantum operation.
That is to say, $\phi (\rho)$ can always be expressed as a finite sum in the
form $\sum_i A_i \rho A_i^\dag$ known as the operator-sum representation with
$\sum_i A_i^\dag A_i = I_{\mathcal H} = \sum_i A_i A_i^\dag$ so that
$\phi(I_{\mathcal H}) = I_{\mathcal H}$, where $I_{\mathcal H}$ denotes the
identity operator on ${\mathcal H}$.
Surely, $A_i$'s (which are called the Kraus operators associated with $\phi$)
as well as elements in ${\mathcal B}({\mathcal H})$ can be regarded as
complex-valued matrices of dimension $\dim {\mathcal H} < \infty$ with respect
to an orthonormal basis of ${\mathcal H}$.
By the fixed convex cone of positive semidefinite operator problem, we refer to
the problem of determining if $\phi$ fixes a convex cone ${\mathcal C}^+$
formed by the set of all positive semidefinite operators in
${\mathcal B}({\mathcal S})$ for a (proper) Hilbert subspace ${\mathcal S}$ of
${\mathcal H}$ in the sense that $\phi[{\mathcal C}^+] \subset {\mathcal C}^+$.
Since $\phi$ is a trace-preserving completely positive map, our fixed convex
cone problem is equivalent to the problem of determining the existence of a
subspace ${\mathcal S}$ of ${\mathcal H}$ such that $\phi$ sends density
matrices in ${\mathcal S}$ to density matrices in ${\mathcal S}$ --- a problem
of quantum information science interest.
(Unlike the approaches in Refs.~\cite{Kribs03,St07,BK08,BK10}, we are not
interested in the action of $\phi$ on $\sigma\in {\mathcal B}({\mathcal H})$
that is not a density matrix because $\sigma$ does not describe a physical
quantum state and hence plays no role in quantum information science.)
By abusing language, we call the Hilbert subspace ${\mathcal S}$ an invariant
subspace of $\phi$; alternatively, we say that $\phi$ fixes the subspace
${\mathcal S}$.
If such a proper subspace ${\mathcal S}$ exists, we would like to explicitly
find it out.

Interestingly, we are going to show in Sec.~\ref{Sec:IIS} that if
${\mathcal S}$ is an invariant subspace of a unital quantum operation $\phi$,
then so is its orthogonal subspace ${\mathcal S}^\perp$.
In other words, a unital quantum operation $\phi$ induces a direct sum
decomposition of the finite-dimensional Hilbert space ${\mathcal H}$ into
irreducible invariant subspaces (\IIS's) of $\phi$.
Here an \IIS means that it does not contain any proper invariant subspace.
More importantly, we further prove in Sec.~\ref{Sec:IIS} that $\phi$ fixes the
subspace ${\mathcal S}$ (and hence also ${\mathcal S}^\perp$) if and only if
${\mathcal S}$ and ${\mathcal S}^\perp$ are simultaneous invariant subspaces of
all Kraus operators $A_i$'s associated with $\phi$.
That is, $A_i [{\mathcal S}] \subset {\mathcal S}$ and $A_i
[{\mathcal S}^\perp] \subset {\mathcal S}^\perp$ for all $i$.
Hence, $\phi$ has an interesting structure in the sense that it induces a
simultaneous block diagonalization for all its Kraus operators $A_i$'s such
that each diagonal block acts on a different irreducible invariant subspace of
$\phi$.
In this regard, our notion of convex cone fixation, which concentrates only on
the action of $\phi$ on density matrices, turns out to be strong enough to
force $\phi[{\mathcal B}({\mathcal S})] \subset {\mathcal B}({\mathcal S})$.

An interesting consequence of this finding is a simple proof of the structural
characterization theorem of unital quantum operation acting on
finite-dimensional density matrices.
As far as we know, we are the first group who explicitly state and prove this
theorem although this theorem can be deduced from the works of
Kribs~\cite{Kribs03} and Blume-Kohout \emph{et al.}~\cite{BK10}.
Furthermore, both prior works did not investigate the quantum information
processing consequences of the structural characterization theorem of
finite-dimensional unital quantum operations.
In fact, Kribs showed in Lemma~2.2 of Ref.~\cite{Kribs03} that fixed points for
an irreducible unital quantum operation acting on ${\mathcal B}({\mathcal H})$
with $\dim {\mathcal H} < \infty$ must be scalars.
This is a special case of Theorem~\ref{Thrm:IS} to be reported in
Sec.~\ref{Sec:IIS}.
However, Kribs did not mention the decomposition of unital quantum operation
into direct sum of irreducible ones and properties of such a decomposition.
Whereas in Lemmas~5.4, 5.5 and Theorem~5 of Ref.~\cite{BK10}, Blume-Kohout
\emph{et al.} used the structure theorem of matrix algebra to prove the
existence of a direct sum decomposition structure for a general
finite-dimensional quantum operation.
They also discussed its implications on the structure of the associated Kraus
operators.
Our structural characterization theorem can be readily deduced from this work
by sharpening their results for the case of a unital operation.
Here we use an alternative approach, which uses rather elementary techniques in
mathematical analysis and graph theory, to obtain the structural
characterization theorem for finite-dimensional unital quantum operations.
In this way, we avoid going through the more technical proofs on a more general
situation in Refs.~\cite{Kribs03,BK10} and adapting them to our particular case
of interest.

The structure theorem reported in Sec.~\ref{Sec:IIS} has a few quantum
information science applications.
In Sec.~\ref{Sec:measurement_statistics}, we first use it to obtain a
simple and intuitive proof of a theorem concerning the calculation of classical
capacity $C_{1,\infty}$ of a unital channel acting on finite-dimensional
quantum states originally obtained by St{\o}rmer in Ref.~\cite{St07}.
Our proof can be used to extend St{\o}rmer's result to the calculation of the
entanglement-assisted classical capacity $C_e$, the coherent information $J$,
and the minimal output $\alpha$-Renyi entropy $S_{\min,\alpha}$ of the same
channel.
More importantly, we completely characterize the kind of quantum states that
are fixed by a unital quantum operation in
Sec.~\ref{Sec:measurement_statistics}.
And we provide a necessary and sufficient condition for the measurement
statistics of a positive operator-valued measure (\POVM) measurement to be
preserved when a general finite-dimensional quantum state passes through a
unital quantum channel.
Finally, we briefly discuss the implications of our results in
Sec.~\ref{Sec:discussions}.

\section{Structural Characterization Of Unital Quantum Operations Acting On
Finite-Dimensional Density Matrices}
\label{Sec:IIS}

We use the following theorem, which was independently proven by Kribs as
Theorem~2.1 in Ref.~\cite{Kribs03} and by Arias \emph{et al.} as Theorem~3.5(a)
in Ref.~\cite{AriasJMP02}, as our starting point to study the existence of
invariant subspace of a unital quantum operation.
(See also the variation of this theorem stated as Lemma~5.2 in
Ref.~\cite{BK10}.)

\begin{theorem}\label{Thrm:Arias}
Let $\phi(\cdot) = \sum_i A_i \cdot A_i^\dag$ be a unital quantum operation
on the set of all linear operators of a finite-dimensional Hilbert space
${\mathcal H}$.
Then, $\phi$ fixes $\sigma\in {\mathcal B}({\mathcal H})$ (that is,
$\phi(\sigma) = \sigma$) if and only if $A_i \sigma=\sigma A_i$ for all $i$. 
\end{theorem}

\begin{remark} \label{Rem:Kraus_operators1}
In particular, Theorem~\ref{Thrm:Arias} relates fixing a quantum state $\rho$
(a property independent of the choice of the associated Kraus operators
$A_i$'s) to the commutativity of $\rho$ with the set of Kraus operators $A_i$'s
used any operator-sum representation of the unital quantum operator $\phi$.
This is possible partly because of a unitary degree of freedom in the
operator-sum representation of $\phi$.
More precisely, $\phi (\cdot) = \sum_i A_i \cdot A_i^\dag = \sum_i B_i \cdot
B_i^\dag$ if and only if $B_j = \sum_i u_{ij} A_i$ for all $j$ where $u_{ij}$
is the $(i,j)$-th element of a unitary matrix.~\cite{MC}
Thus, the commutativity of $\rho$ with all $A_i$'s implies the commutativity of
$\rho$ with all $B_i$'s.
\end{remark}

A special case of Theorem~\ref{Thrm:Arias} is that $\phi$ fixes a
(normalized) pure state $|x\rangle$ if and only if $|x\rangle$ is an
eigenvector of $A_i$ for all $i$.
(The ``if part'' can be deduced by the fact that $A_i|x\rangle = \lambda_i
|x\rangle$ implies $\phi(|x\rangle\langle x|) = \sum_i |\lambda_i|^2
|x\rangle\langle x| = |x\rangle\langle x|$ for $\phi$ is trace-preserving.
The ``only if part'' follows from
\begin{equation}
 A_i |x\rangle = A_i |x\rangle\langle x|x\rangle
 = |x\rangle\langle x| A_i |x\rangle
 = \langle x|A_i|x\rangle\,|x\rangle . \label{E:fixed_pure_state_impl}
\end{equation}
Note also that by the same argument in Remark~\ref{Rem:Kraus_operators1},
$|x\rangle$ is an eigenvector for any Kraus operators used in the operator-sum
representation of $\phi$ although the corresponding eigenvalues are
operator-sum representation dependent.)
In other words, $\phi$ fixes an one-dimensional subspace ${\mathcal S}$ of
${\mathcal H}$ if and only if ${\mathcal S}$ is a simultaneous eigenspace of
all its Kraus operators $A_i$'s.

We now apply Theorem~\ref{Thrm:Arias} to study the necessary and sufficient
condition for the existence of invariant subspace of $\phi$.

\begin{theorem}\label{Thrm:IS}
Let $\phi$ be a unital quantum operation on the set of all linear operators of
a finite-dimensional Hilbert space ${\mathcal H}$.
Then, $\phi$ fixes a Hilbert subspace ${\mathcal S}$ of ${\mathcal H}$ if and
only if $\phi(P_{\mathcal S}) = P_{\mathcal S}$ where $P_{\mathcal S}$ is the
projection operator onto ${\mathcal S}$.
\end{theorem}

\begin{proof}
($\Leftarrow$):
Assume there exists an Hermitian operator $\sigma_{\mathcal{S}}$ whose support
is in ${\mathcal S}$.
Suppose further that $0 \leq \sigma_{\mathcal S}\leq P_{\mathcal S}$ and
the support of $\phi(\sigma_{\mathcal S})$ does not belong to ${\mathcal S}$.
Then, $\sigma_{{\mathcal S}'} \equiv P_{\mathcal S} - \sigma_{\mathcal S} \geq
0$.
More importantly,
\begin{equation}
P_{\mathcal S} = \phi(P_{\mathcal S}) = \phi(\sigma_{\mathcal S}) + \phi
(\sigma_{{\mathcal S}'}) .
\label{E:IS_proof1}
\end{equation}
Since $\phi(\sigma_{\mathcal S})\not\in \mathcal{S}$, there exists $|x\rangle
\in {\mathcal S}^\perp$ such that $\langle x|\phi(\sigma_{\mathcal{S}})
|x\rangle\neq 0$ where ${\mathcal S}^\perp$ denotes the orthogonal complement
of ${\mathcal S}$.
And by the positivity of $\phi(\sigma_{\mathcal{S}})$, we have $\langle x| \phi
(\sigma_{\mathcal{S}})|x\rangle> 0$.
But then 
\begin{align}
\langle x|\phi(\sigma_{\mathcal{S}'})|x\rangle
&=\langle x|P_{\mathcal S}|x\rangle - \langle x|\phi(\sigma_{\mathcal{S}})
 |x\rangle\nonumber\\
&=-\langle x|\phi(\sigma_{\mathcal{S}})|x\rangle < 0 . \label{E:IS_proof2}
\end{align}
This is impossible for $\phi$ is a quantum operation.
Therefore, we conclude that $\phi(\sigma_{\mathcal S}) \in {\mathcal S}$ for
all non-negative operators on ${\mathcal S}$.
In other words, $\phi$ fixes ${\mathcal S}$.

($\Rightarrow$):
$\phi$ fixes $\mathcal{S}$ implies $\phi(P_{\mathcal S})\in {\mathcal S}$. 
Assume the contrary is true so that $\phi(P_{\mathcal S}) \neq P_{\mathcal S}$.
Let $P_{{\mathcal{S}}^{\perp}}=P_{\mathcal H} - P_{\mathcal S}$ be the
projection operator onto the orthogonal complement of ${\mathcal S}$.
Then,
\begin{equation}
\phi(P_{\mathcal S})+\phi(P_{{\mathcal{S}}^{\perp}}) = \phi(P_{\mathcal H}) =
P_{\mathcal H} . \label{E:IS_proof3}
\end{equation}
Note that the last equality in the above equation follows from the fact that
$\phi$ is unital.
Since ${\mathcal H}$ is finite-dimensional and $\phi$ fixes ${\mathcal S}$, we
can express $\phi(P_{\mathcal S})$ as the finite sum $\sum_{j=1}^{\dim
{\mathcal S}} a_j|y_j\rangle\langle y_j|$ with non-negative $a_j$'s, where $\{
|y_j\rangle \}$ is an orthonormal basis of ${\mathcal S}$.
Since $\phi(P_{\mathcal S})\neq P_{\mathcal S}$, we cannot have $a_j = 1$ for
$j = 1,2,\dots, \dim {\mathcal S}$.
Nevertheless, since $\phi$ is trace-preserving, $a_j$'s still have to satisfy
the constraint $\sum_{j=1}^{\dim {\mathcal S}} a_j = \dim {\mathcal S}$.
Thus, by relabeling the index if necessary, we may assume that $a_1>1$ and
$a_2<1$.
So, from Eq.~\eqref{E:IS_proof3},
\begin{align}
1&= \langle y_1|\phi(P_{\mathcal S})|y_1\rangle + \langle y_1|\phi
 (P_{{\mathcal S}^\perp})|y_1\rangle\nonumber\\
&=a_1+\langle y_1|\phi(P_{{\mathcal S}^\perp})|y_1\rangle . \label{E:IS_proof4}
\end{align}
That is to say, $\langle y_1|\phi(P_{{\mathcal S}^\perp})|y_1\rangle = 1-a_1 <
0$, which contradicts the assumption that $\phi$ is a quantum operation.
Therefore, we conclude that $\phi(P_{\mathcal S}) = P_{\mathcal S}$.
\end{proof}

\begin{remark}\label{Rem:role_of_unital}
Note that the unital condition is needed only in the proof of the ``only if
part'' of the Theorem.
Note further that Kribs proved the special case of this theorem when $\phi$
does not fix any proper subspace of ${\mathcal H}$~\cite{Kribs03}.
\end{remark}

\begin{corollary}\label{Cor:decomposition}
$\phi$ fixes ${\mathcal S}$ if and only if $\phi$ fixes ${\mathcal S}^\perp$.
\end{corollary}

\begin{proof}
By Theorem~\ref{Thrm:IS} and Eq.~\eqref{E:IS_proof3}, $\phi$ fixes
${\mathcal S}$ if and only if $\phi (P_{{\mathcal S}^\perp}) = P_{\mathcal H} -
\phi(P_{\mathcal S}) = P_{\mathcal H} - P_{\mathcal S} =
P_{{\mathcal S}^\perp}$, which in turn is true if and only if $\phi$ fixes
${\mathcal S}^\perp$. 
\end{proof}

\begin{corollary}\label{Cor:fix_the_subspace}
$\phi$ fixes the subspace $\mathcal{S}$ if and only if
\begin{equation}
\langle s^\perp|A_i|s\rangle = 0 = \langle s|A_i|s^\perp\rangle
\label{E:fix_the_subspace_statement1}
\end{equation}
for all $|s\rangle\in \mathcal{S}$ and $|s^\perp\rangle\in \mathcal{S}^\perp$
and for all Kraus operators $A_i$'s of $\phi$.
\end{corollary}

\begin{proof}
From Theorem~\ref{Thrm:IS}, $\phi$ fixes ${\mathcal S}$ implies $\phi
(P_{\mathcal S}) = P_{\mathcal S}$.
Theorem~\ref{Thrm:Arias} further implies $A_i P_{\mathcal S}= P_{\mathcal S}
A_i$ for all $i$.
Multiplying $\langle s^\perp|$ on the left and $|s\rangle$ on the right gives
$\langle s^\perp|A_i|s\rangle=0$.
Similarly, multiplying $\langle s|$ on the left and $|s^\perp\rangle$ on the
right gives $\langle s|A_i|s^\perp\rangle=0$.

To prove the converse, one only needs to observe that
Eq.~\eqref{E:fix_the_subspace_statement1} implies $A_i [{\mathcal S}] \subset
{\mathcal S}$ for all $i$.
Thus, $\phi$ fixes ${\mathcal S}$.
\end{proof}

\begin{remark}\label{Rem:fix_the_subspace}
Surprisingly, the notion of convex cone fixation by a unital quantum operation
$\phi$ is much stronger than what we have originally written down.
Recall from Sec.~\ref{Sec:intro} that $\phi$ fixes a Hilbert subspace
${\mathcal S}$ simply means $\phi[C_{\mathcal S}^+] \subset C_{\mathcal S}^+$
where $C_{\mathcal S}^+$ denotes the convex cone formed by the set of all
positive semidefinite operators in ${\mathcal B}({\mathcal S})$.
Yet, the above two Corollaries say that subspace fixing actually demands
something more, namely, $\phi [{\mathcal B}({\mathcal S})] \subset
{\mathcal B}({\mathcal S})$, $\phi [{\mathcal B}({\mathcal S}^\perp)] \subset
{\mathcal B}({\mathcal S}^\perp)$.
Note further that we cannot directly deduce
Corollary~\ref{Cor:fix_the_subspace} from Corollary~\ref{Cor:decomposition} for
the latter says nothing on $\phi(\sigma)$ for a general operator $\sigma$ whose
support is in ${\mathcal S}$.
\end{remark}

This motivates us to formulate our main theorem.

\begin{theorem}[Structural theorem for unital quantum operations on
finite-dimensional density matrices]
\label{Thrm:structural_theorem}
(a)~
Every finite-dimensional unital quantum operation $\phi\colon
{\mathcal B}({\mathcal H}) \to {\mathcal B}({\mathcal H})$ induces a direct sum
decomposition of ${\mathcal H} = \bigoplus_j {\mathcal S}_j$ where each
${\mathcal S}_j$ is an \IIS of $\phi$.
Furthermore, every Kraus operator $A_i$ of $\phi$ can also be decomposed as
$\bigoplus_j A_i^{{\mathcal S}_j}$ where $A_i^{{\mathcal S}_j} \in
{\mathcal B}({\mathcal S}_j)$ for all $i,j$.
Hence, the quantum operation $\left. \phi \right|_{{\mathcal B}({\mathcal S})}$
is also unital whenever ${\mathcal S}$ is an invariant subspace of $\phi$.
In other words, $\phi$ can be expressed as a direct sum $\bigoplus_j \left.
\phi \right|_{{\mathcal B}({\mathcal S}_j)}$ of unital quantum operations.
In matrix language, $\phi$ has a proper invariant subspace if and only if each
of its Kraus operators can be simultaneously block diagonalized by unitary
conjugation into at least two diagonal blocks.

\noindent (b)~
In addition, ${\mathcal S}$ is an \IIS of $\phi$ if and only if all fixed
positive self-adjoint operators of $\phi$ in ${\mathcal B}({\mathcal S})$ are
in the form $a P_{\mathcal S}$ for some $a\geq 0$, where $P_{\mathcal S}$
denotes the projection operator onto ${\mathcal S}$.
\end{theorem}

\begin{proof}
From Corollaries~\ref{Cor:decomposition} and~\ref{Cor:fix_the_subspace}, every
Kraus operator $A_i$ of a unital quantum operator $\phi\colon
{\mathcal B}({\mathcal H}) \to {\mathcal B}({\mathcal H})$ with $\dim
{\mathcal H} < \infty$ admits a direct sum decomposition $A_i^{\mathcal S}
\oplus A_i^{{\mathcal S}^\perp}$ where $A_i^{\mathcal S}$
($A_i^{{\mathcal S}^\perp}$) is a linear operator in the invariant subspace
${\mathcal S}$ (${\mathcal S}^\perp$) of $\phi$.
Since $\sum_i A_i^{\mathcal S} {A_i^{\mathcal S}}^\dag = I_{\mathcal S}$,
the quantum operation $\left. \phi \right|_{{\mathcal B}({\mathcal S})}
(\cdot) \equiv \sum_i A_i^{\mathcal S} \cdot {A_i^{\mathcal S}}^\dag$ is
unital.
By the same token, $\left. \phi \right|_{{\mathcal B}({\mathcal S}^\perp)}
(\cdot) \equiv \sum_i A_i^{{\mathcal S}^\perp} \cdot
{A_i^{{\mathcal S}^\perp}}^\dag$ is unital.
By recursively applying Corollaries~\ref{Cor:decomposition}
and~\ref{Cor:fix_the_subspace} to the unital quantum operations
$\left. \phi \right|_{{\mathcal B}(\mathcal S)}(\cdot)$ and $\left. \phi
\right|_{{\mathcal B}({\mathcal S}^\perp)}$ at most $\dim \left( {\mathcal H}
\right) - 1$ times, part~(a) is proven.

To prove part~(b), suppose $0\leq \sigma \in {\mathcal B}({\mathcal S})$ is
fixed by $\phi$ and yet $\sigma \neq a P_{\mathcal S}$ for all $a\geq 0$.
Since ${\mathcal H}$ and hence ${\mathcal S}$ are finite-dimensional, we may
write $\sigma = \sum_{j=1}^{\dim {\mathcal S}} b_j |y_j\rangle\langle y_j|$ for
some orthonormal basis vectors $|y_j\rangle$'s of ${\mathcal S}$ and all
$b_j$'s are non-negative.
By relabeling the indices if necessary, may we assume that $b_1 \leq b_j$
for all $j$; and the requirement that $\sigma\neq a P_{\mathcal S}$ implies not
all $b_j$'s are equal.
Using part~(a), $\left. \phi \right|_{\mathcal S}$ is also unital so that
$\phi$ fixes $P_{\mathcal S}$ and hence also the operator $\sigma'$ given by
\begin{equation}
\sigma' = \sigma - b_1 P_{\mathcal S} 
= \sum_{j=2}^{\dim {\mathcal S}} \left( b_j - b_1 \right) |y_j\rangle\langle
y_j| > 0 . \label{E:structural_proof1}
\end{equation}
From Theorem~\ref{Thrm:Arias}, $A_i^{\mathcal S} \sigma' = \sigma'
A_i^{\mathcal S}$ for all $i$.
By multiplying $\langle y_1|$ on the left and $|y_j\rangle$ on the right, we
get $(b_j-b_1) \langle y_1|A_i^{\mathcal S}|y_j\rangle = 0$ for all $i,j$.
Similarly, multiplying $\langle y_j|$ on the left and $|y_1\rangle$ on the
right gives $(b_j-b_1) \langle y_j|A_i^{\mathcal S}|y_1\rangle = 0$ for all
$i,j$.
That is, $\langle y_j|A_i^{\mathcal S}|y_1\rangle = \langle y_1|
A_i^{\mathcal S}|y_j\rangle = 0$ for all $i$ whenever $b_1 < b_j$.
Following the same logic, we conclude that
\begin{equation}
\langle y_j|A_i^{\mathcal S} |y_k\rangle = \langle y_k|A_i^{\mathcal S}
|y_j\rangle = 0 \label{E:structural_proof2}
\end{equation}
for all $i$ whenever $b_1 = b_k < b_j$.
Consequently, each of the $A_i^{\mathcal S}$'s can be simultaneously block
diagonalized by unitary conjugation to at least two diagonal blocks --- one
block corresponds to those indices $k$'s with $b_k = b_1$ and the other block
corresponds to those $k$'s with $b_k > b_1$.
From Corollary~\ref{Cor:fix_the_subspace}, ${\mathcal S}$ is not an \IIS of
${\mathcal H}$, which is absurd.

Finally, to show the converse of part~(b), suppose ${\mathcal S}'$ is a proper
subspace of ${\mathcal S}$ that is fixed by $\phi$.
Then, it is clear from part~(a) that $\phi$ fixes the projection operator
$P_{{\mathcal S}'}$.
Hence, not all positive self-adjoint operators in ${\mathcal B}({\mathcal S})$
fixed by $\phi$ are in the form $a P_{\mathcal S}$.
This proves the theorem.
\end{proof}

\begin{remark}\label{Rem:prior_art}
As we have mentioned in Sec.~\ref{Sec:intro},
Theorem~\ref{Thrm:structural_theorem} can also be deduced from the work of
Kribs~\cite{Kribs03} and Blume-Kohout \emph{et al.}~\cite{BK10}.
\end{remark}

\begin{remark} \label{Rem:Kraus_operators2}
The unitary degree of freedom in operator-sum representation mentioned in
Remark~\ref{Rem:Kraus_operators1} is the reason why
Eq.~\eqref{E:fix_the_subspace_statement1} in
Corollary~\ref{Cor:fix_the_subspace} as well as the simultaneous block
diagonalization of $A_i$'s in part~(a) of Theorem~\ref{Thrm:structural_theorem}
hold for any Kraus operators associated with $\phi$.
\end{remark}

\begin{remark}\label{Rem:non-essential}
Actually, positivity of the self-adjoint operator $\sigma$ is not essential in
the proof of part~(b) of Theorem~\ref{Thrm:structural_theorem}.
In fact, it can be slightly strengthened as ${\mathcal S}$ is an \IIS of $\phi$
if and only if all fixed self-adjoint operators of $\phi$ in
${\mathcal B}({\mathcal H})$ are in the form $a P_{\mathcal S}$ for some $a\in
{\mathbb R}$.
The proof is left to interested readers.
We do not state this slightly more general form in the theorem because we are
more interested in the action of $\phi$ to a physical quantum state.
\end{remark}

\begin{remark} \label{Rem:non-uniqueness_decomposition}
In general, a matrix can be diagonalized into irreducible blocks in more than
one orthonormal basis due to degeneracy of its eigenspace.
This is also the reason why the \IIS decomposition of a finite-dimensional
unital quantum operation discussed in the above theorem need not be unique.
\end{remark}

The following example illustrates the power of
Theorem~\ref{Thrm:structural_theorem} in determining the structures of some
well-known unital quantum channels.
\begin{example} \label{Eg:depolarizing_channel}
Consider the depolarization qudit channel $\phi$ over a finite-dimensional
Hilbert space ${\mathcal H}$.
That is, $\phi$ sends a density matrix $\rho$ to $(1-p) \rho + p I_{\mathcal H}
/ \dim {\mathcal H}$ with $0 < p \leq 1$.
More generally, $\phi(\sigma) = (1-p) \sigma + p I_{\mathcal H} \Tr (\sigma) /
\dim {\mathcal H}$ for all $\sigma\in {\mathcal B}({\mathcal H})$.
Let $P_{\mathcal S}$ be the projector on a Hilbert subspace ${\mathcal S}$ of
${\mathcal H}$.
Then, $\phi$ fixes $P_{\mathcal S}$ if and only if $p P_{\mathcal S} = p
I_{\mathcal H} \dim {\mathcal S} / \dim {\mathcal H}$.
As $p \neq 0$, $\phi$ fixes $P_{\mathcal S}$ if and only if ${\mathcal S} =
{\mathcal H}$.
Since $\phi$ is unital, Theorem~\ref{Thrm:structural_theorem} implies that
${\mathcal H}$ is the only \IIS of $\phi$.
In this regard, all finite-dimensional depolarizing qudit channels cannot be
further decomposed as a direct sum of two unital quantum operations.
\end{example}

Although \IIS decomposition induced by $\phi$ may not be unique, different
\IIS decompositions share the same dimensional structure.
This is analogous to the dimensions of \IIS for complex-valued matrices.

\begin{theorem} \label{Thrm:block_invariance}
Suppose ${\mathcal H} = \bigoplus_{i=1}^m {\mathcal S}_i = \bigoplus_{j=1}^n
{\mathcal S}'_j$ be two different \IIS decompositions induced by $\phi$
reported in part~(a) of Theorem~\ref{Thrm:structural_theorem} indexed in such a
way that $\dim {\mathcal S}_i \geq \dim {\mathcal S}_{i'}$ whenever $i > i'$
and $\dim {\mathcal S}'_j \geq \dim {\mathcal S}'_{j'}$ whenever $j > j'$.
Then $m = n$ and $\dim {\mathcal S}_i = \dim {\mathcal S}'_i$ for all $i$.
\end{theorem}

\begin{proof}
Since there are two distinct \IIS decompositions for $\phi$, we can always find
an \IIS in the first decomposition, say, ${\mathcal S}_1$ that is distinct from
all the \IIS's in the second decomposition.
As ${\mathcal S}_1 \subset \bigoplus_j {\mathcal S}'_j$, we can find an \IIS in
the second decomposition, say ${\mathcal S}'_1$ which is not contained in
${\mathcal S}^\perp_1$ the orthogonal subspace of ${\mathcal S}_1$.

We claim that $\dim {\mathcal S}_1 = \dim {\mathcal S}'_1$.
Suppose the contrary, may we assume without lost of generality that $\dim
{\mathcal S}_1 > \dim {\mathcal S}'_1$.
Let us write ${\mathcal S}_1 = {\mathcal T}_1 \oplus {\mathcal U}_1$ and
${\mathcal S}'_1 = {\mathcal T}'_1 \oplus {\mathcal U}'_1$ where
${\mathcal T}_1 = {\mathcal S}_1 \cap {{\mathcal S}'}^\perp_1$ and
${\mathcal T}'_1 = {\mathcal S}'_1 \cap {\mathcal S}^\perp_1$.
Since $\dim {\mathcal S}_1 > \dim {\mathcal S}'_1$ and ${\mathcal S}'_1 \not
\subset {\mathcal S}^\perp_1$, we conclude that ${\mathcal T}_1$ is a proper
subspace of ${\mathcal S}_1$.

Because ${\mathcal S}_1$ and ${\mathcal S}'_1$ are \IIS's of $\phi$, by
part~(b) of Theorem~\ref{Thrm:structural_theorem}, $\phi$ fixes the projectors
$P_{{\mathcal S}_1}$, $P_{{\mathcal S}'_1}$ and hence also the positive
operator $\sigma = P_{{\mathcal S}_1} + 0.5 P_{{\mathcal S}'_1} \in
{\mathcal B}({\mathcal S}_1 + {\mathcal S}'_1)$.

We are going to show that the only eigenvectors of $\sigma$ with eigenvalue $1$
are vectors in the subspace ${\mathcal T}_1$.
A vector in ${\mathcal S}_1 + {\mathcal S}'_1$ can be uniquely written as
$|\psi\rangle = a|x\rangle + b|x^\perp\rangle$ where $|x\rangle$ and
$|x^\perp\rangle$ are normalized vectors in ${\mathcal S}_1$ and $\left(
{\mathcal S}_1 + {\mathcal S}'_1 \right) \cap {\mathcal S}^\perp_1$,
respectively.
Then $|\psi\rangle$ is an eigenvector of $\sigma$ with eigenvalue $1$ provided
that
\begin{equation}
\left( P_{{\mathcal S}_1} + 0.5 P_{{\mathcal S}'_1} \right) \left( a|x\rangle +
b|x^\perp\rangle \right) = a|x\rangle + b|x^\perp\rangle .
\label{E:block_invariance_proof1}
\end{equation}
Multiplying $\langle x|$ and $\langle x^\perp|$ to
Eq.~\eqref{E:block_invariance_proof1} gives the system of equations
\begin{subequations}
\begin{empheq}[left=\empheqlbrace]{align}
a \langle x|P_{{\mathcal S}'_1}|x\rangle &= -b \langle x|
 P_{{\mathcal S}'_1}|x^\perp\rangle , \label{E:block_invariance_proof2a}
\\
0.5 a \langle x^\perp|P_{{\mathcal S}'_1}|x\rangle &= b \left( 1 - 0.5
 \langle x^\perp|P_{{\mathcal S}'_1}|x^\perp\rangle \right) .
 \label{E:block_invariance_proof2b}
\end{empheq}
\end{subequations}
This system of equations has non-trivial solution if and only if
\begin{equation}
\langle x|P_{{\mathcal S}'_1}|x\rangle \left( 1 - 0.5 \langle x^\perp|
P_{{\mathcal S}'_1}|x^\perp\rangle \right) = -0.5 \left| \langle x^\perp|
P_{{\mathcal S}'_1}|x\rangle \right|^2 . \label{E:block_invariance_proof3}
\end{equation}
As $P_{{\mathcal S}'_1}$ is a projector, $\langle y|P_{{\mathcal S}'_1}
|y\rangle \in [0,1]$ for all normalized vector $|y\rangle$.
Therefore, the L.H.S. and R.H.S. of Eq.~\eqref{E:block_invariance_proof3} are
non-negative and non-positive, respectively.
Consequently, both the R.H.S. and R.H.S. of
Eq.~\eqref{E:block_invariance_proof3} equal $0$.
Thus, $\langle x^\perp |P_{{\mathcal S}_1'}|x\rangle = 0$.
Since $1-0.5\langle x^\perp|P_{{\mathcal S}'_1}|x^\perp\rangle > 0$,
Eq.~\eqref{E:block_invariance_proof2b} implies $b = 0$.
So, $|\psi\rangle = a|x\rangle \in {\mathcal T}_1$ is the solution of
Eq.~\eqref{E:block_invariance_proof1}.
In other words, we have succeeded in showing that all eigenvectors of $\sigma$
with eigenvalue $1$ are the vectors in ${\mathcal T}_1$.

Therefore, $\sigma$ can be rewritten as $P_{{\mathcal T}_1} \oplus \sigma'$
where $\sigma'$ is a positive operator whose eigenspectrum does not contain the
eigenvalue $1$ (and hence its support is not in ${\mathcal T}_1$).
Using the argument in the proof of part~(b) of
Theorem~\ref{Thrm:structural_theorem}, we conclude that $\phi$ fixes
$P_{{\mathcal T}_1}$.
In other words, ${\mathcal T}_1$ is an proper invariant subspace of
${\mathcal S}_1$.
This contradicts the fact that ${\mathcal S}_1$ is irreducible.
Therefore, we conclude that $\dim {\mathcal S}_1 = \dim {\mathcal S}'_1$.

To summarize, so far we have deduced that if ${\mathcal S}_i \not\subset
{{\mathcal S}'}^\perp_j$ (and hence ${\mathcal S}'_j \not\subset
{\mathcal S}^\perp_i$), then $\dim {\mathcal S}_i = \dim {\mathcal S}'_j$.
We now construct a finite bipartite graph ${\mathfrak G}$ whose vertices are
the \IIS's ${\mathcal S}_i$'s and ${\mathcal S}'_j$'s --- each side
corresponds to a different \IIS decomposition of ${\mathcal H}$.
An edge is drawn between vertices ${\mathcal S}_i$ and ${\mathcal S}'_j$
each picked from one side provided that either (i)~${\mathcal S}_i \not \subset
{{\mathcal S}'}^\perp_j$ or (ii)~${\mathcal S}_i = {\mathcal S}'_j$.
Our construction of ${\mathfrak G}$ guarantees that:
\begin{itemize}
\item Each \IIS must connect to at least one \IIS on the other side of the
graph.
\item The dimension of each Hilbert space associated with the \IIS's in the
same connected component of ${\mathfrak G}$ agrees.
\item The direct sum of \IIS's drawn from each of the two sides in a given
connect component of ${\mathfrak G}$ must agree.
\end{itemize}
Since ${\mathcal H}$ is finite-dimensional, we further conclude that in each
connected component, the number of \IIS's in the two sides of ${\mathfrak G}$
are the same.
Therefore, in each connected component, we can construct a bijection from
\IIS's in one side to the \IIS's in the other side of the graph
${\mathfrak G}$.
This bijection maps each \IIS in one of the \IIS decomposition to an \IIS of
equal dimension in the other decomposition.
Hence the theorem is proved.
\end{proof}

\begin{remark} \label{Rem:graphi_algorithm_complexity}
Actually, the construction of the bipartite graph and the determination of the
connected components of this graph in the proof of
Theorem~\ref{Thrm:block_invariance} are quite efficient.
To see this, first by Gram-Schmidt orthogonalization procedure, we may
represent each subspace ${\mathcal S}_i$ of ${\mathcal H}$ by its orthogonal
basis.
Clearly, the computational cost needed to specify all the orthogonal bases of
${\mathcal S}_i$'s and ${\mathcal S}'_j$'s, measured in terms of the number of
real number arithmetical operations, scales like $\text{O}((\dim
{\mathcal H})^2)$.
Given a normalized vector $|v\rangle$ and a subspace ${\mathcal S}$ specified
by a set of orthonormal basis $\{ |b_k\rangle \}$, we may determine if
$|v\rangle$ is in ${\mathcal S}$ or not by checking if $\sum_i \left| \langle
v|b_i \rangle \right|^2$ equal to $1$ or not within a certain error tolerant
level $\epsilon$ caused by finite-precision arithmetic.
Similarly, we know if $|v\rangle$ is in ${\mathcal S}^\perp$ or not by checking
if $\sum_i \left| \langle v|b_i\rangle \right|^2$ is equal to $0$ or not within
a certain error tolerant level $\epsilon$.
Both procedures require $\text{O}(\dim {\mathcal S} \ \dim {\mathcal H})$ real
number arithmetical operations.
Using these procedures as subroutines, we can determine if ${\mathcal S}_i \not
\subset {{\mathcal S}'}_j^\perp$ or ${\mathcal S}_i = {\mathcal S}'_j$ in
$\text{O}(\dim {\mathcal S}_i \ \dim {\mathcal S}'_j \ \dim {\mathcal H})$ time
by checking if each basis vector of ${\mathcal S}$ is in
${{\mathcal S}'}_j^\perp$ or ${\mathcal S}'_j$.
Hence, the bipartite graph ${\mathfrak G}$ described in the proof of
Theorem~\ref{Thrm:block_invariance} can be constructed in $\sum_{i,j} \text{O}
(\dim S_i \ \dim {\mathcal S}'_j \ \dim {\mathcal H}) = \text{O}((\dim
{\mathcal H})^3)$ time.
Furthermore, connected components of a graph can be found by the well-known
depth first search algorithm in graph theory, which needs a computational time
linear in the number of edges of the
graph~\cite{[{See, for example, }][]Graph05}.
So, in our case, connected components of ${\mathfrak G}$ can be determined in
at most $\text{O}((\dim {\mathcal H})^2)$ time.
To summarize, the algorithm concerning the graph ${\mathfrak G}$ in the proof
of Theorem~\ref{Thrm:block_invariance} takes $\text{O}((\dim {\mathcal H})^3)$
time with the graph construction being the most time-consuming step.
Thus, the algorithm of identifying isomorphic \IIS decompositions reported in
the proof of Theorem~\ref{Thrm:block_invariance} is quite efficient and can be
used in actual situations especially when $\dim {\mathcal H}$ is large.
\end{remark}

\section{Simple applications of the structural theorem in quantum information
processing}
\label{Sec:measurement_statistics}

We report three simple applications of Theorem~\ref{Thrm:structural_theorem}
that are of quantum information science interest.
The first one concerns a simplified proof and an extension of the result by
St{\o}rmer in Ref.~\cite{St07} on the channel capacity of a unital channel
acting on finite-dimensional quantum states.
The second one concerns the type of quantum states that are invariant after
passing through a unital quantum channel.
The third one concerns the preservation of \POVM measurement statistics under a
unital quantum channel.
Our work here strengthens the findings of Arias \emph{et al.} in
Ref.~\cite{AriasJMP02}.

In this Section, $\phi$ always denotes a unital quantum operation on
${\mathcal B}({\mathcal H})$ with Kraus operators $A_i$'s and $\dim
{\mathcal H} < \infty$.

\begin{theorem} \label{Thrm:capacities}
Let ${\mathcal H} = \bigoplus_j {\mathcal S}_j$ be a direct sum decomposition
of the Hilbert space ${\mathcal H}$ into \IIS's induced by $\phi$.
Then
\begin{subequations}
\begin{itemize}
 \item the minimal output $\alpha$-Renyi entropy $(\alpha \ge 1)$ is given by
  \begin{equation}
   S_{\min,\alpha}(\phi) = \min_j \left[ S_{\min,\alpha} \left( \left. \phi
   \right|_{{\mathcal B}({\mathcal S}_j)} \right) \right] ,
   \label{E:Renyi_entropy}
  \end{equation}
 \item the coherent information is given by
  \begin{equation}
   J(\phi) = \max_j \left[ J\left( \left. \phi
   \right|_{{\mathcal B}({\mathcal S}_j)} \right) \right] ,
   \label{E:coherent_inform}
  \end{equation}
 \item the entanglement-assisted classical capacity is given by
  \begin{equation}
   C_e(\phi) = \log \sum_j 2^{C_e ( \left. \phi
   \right|_{{\mathcal B}({\mathcal S}_j)})} ,
   \label{E:ent_ass_capacity}
  \end{equation}
 \item the (unassisted) classical capacity is given by
  \begin{equation}
   C_{1,\infty}(\phi) = \log \sum_j 2^{C_{1,\infty} ( \left. \phi
   \right|_{{\mathcal B}({\mathcal S}_j)})} .
   \label{E:classical_capacity}
  \end{equation}
\end{itemize}
\end{subequations}
\end{theorem}

\begin{proof}
By Theorem~\ref{Thrm:structural_theorem}, we know that the $\phi$ can be
expressed as a direct sum of finite-dimensional quantum operations.
Then, all the above formulae concerning the reduction of various
information-theoretic quantities of the quantum operation $\phi$ follow from
Proposition~1 in Ref.~\cite{FW07} by Fukuda and Wolf, which used basic
properties such as von~Neumann entropy and mutual information are concave
functions to show that each of the above capacities can be achieved by a
density matrix respecting the direct sum structure of the quantum operation
$\phi$.
\end{proof}

\begin{remark} \label{Rem:capacities}
Eq.~\eqref{E:classical_capacity} was first proven by St{\o}rmer in Theorem~3 of
Ref.~\cite{St07}.
The rather complicated conditions in that Theorem is nothing but the structural
decomposition of a unital $\phi$ according to our
Theorem~\ref{Thrm:structural_theorem}(b).
\end{remark}

\begin{corollary} \label{Cor:fixed_state_classification}
The set of all quantum states that are fixed by $\phi$ are the classical
mixtures of the completely mixed states $\left( \dim {\mathcal S_j}
\right)^{-1} P_{{\mathcal S}_j}$ where each ${\mathcal S}_j$ is an \IIS of
$\phi$.
\end{corollary}

\begin{proof}[Outline proof]
We can always write a density matrix $\rho\in {\mathcal B}({\mathcal H})$ in
the form $\sum_{j=1}^{\dim {\mathcal H}} b_j |y_j\rangle\langle y_j|$ for some
orthonormal basis vectors $|y_i\rangle$'s of ${\mathcal H}$ with $b_1 \leq b_j
\leq 0$ for all $j$.
Suppose $\phi(\rho) = \rho$.
Following the same argument used in the proof of part~(b) of
Theorem~\ref{Thrm:structural_theorem}, we know that $\rho - b_1 P_{\mathcal H}$
is a positive operator fixed by $\phi$.
Besides, ${\mathcal T} = \Span \{ |y_j\rangle \colon b_j = b_1 \}$ and
${\mathcal T}' = \Span \{ |y_j\rangle \colon b_j \neq b_1 \}$ are two mutually
orthogonal invariant subspaces of $\phi$.
As ${\mathcal H}$ is finite-dimensional, by recursively applying this argument
to $\rho - b_1 P_{\mathcal H}$ and $\left. \phi \right|_{{\mathcal T}'}$ a
finite number of times until $\rho - b_1 P_{\mathcal H} = 0$, we conclude that
$\rho$ is a finite sum in the form $\sum_j c_j P_{{\mathcal T}_j}$ where $c_j
\geq 0$ and ${\mathcal T}_j$'s are mutually orthogonal invariant subspaces of
${\mathcal H}$.
Since each $P_{{\mathcal T}_j}$ is (possibly) a classical mixture of completely
mixed states in the form $\left( \dim {\mathcal S}_k \right)^{-1}
P_{{\mathcal S}_k}$'s, this corollary is proved.
\end{proof}

The above corollary means that completely mixed states of \IIS's are the basic
building blocks of quantum states fixed by $\phi$.
In this sense, the problem of finding and characterizing invariant subspaces of
$\phi$ we are studying here is more than a variation of the fixed state problem
of $\phi$.
It is, in fact, a generalization of the fixed state problem.

\begin{remark}\label{Rem:extended_fixed_state_classification}
Actually, the self-adjointness of $\rho$ is essential in the proof of
Corollary~\ref{Cor:fixed_state_classification} while the positivity of $\rho$
is not.
So, the Corollary can be slightly extended by saying that all self-adjoint
operators that are fixed by $\phi$ must be in the form $\sum_i a_i
P_{{\mathcal S}_i}$ with $a_i\in {\mathbb R}$.
Nevertheless, we stress that Corollary~\ref{Cor:fixed_state_classification}
does not cover the case of fixing a non-self-adjoint operator --- a situation
of no physical meaning.
\end{remark}

\begin{remark}\label{Rem:fix_degenerate_states}
Using Theorem~\ref{Thrm:structural_theorem},
Corollary~\ref{Cor:fixed_state_classification} and the fact that a matrix
admits two non-orthogonal eigenvectors if and only if it has a degenerate
eigenspace, it is easy to see that the following statements are equivalent:
\begin{itemize}
\item $\phi$ fixes two non-orthogonal pure states $|x_1\rangle, |x_2\rangle$.
\item $\phi$ admits two distinct \IIS decompositions such that the first
decomposition contains the \IIS generated by $|x_1\rangle$ and the second
decomposition contains the \IIS generated by $|x_2\rangle$.
\item $|x_1\rangle, |x_2\rangle$ are degenerate eigenvectors of each of the
Kraus operators $A_i$'s of $\phi$ and that $A_i P_{\mathcal S} = P_{\mathcal S}
A_i$ for all $i$ where $P_{\mathcal S}$ denotes the projector onto the space
spanned by $|x_1\rangle$ and $|x_2\rangle$.
\item $\phi$ fixes all pure states in the span of $|x_1\rangle, |x_2\rangle$.
\end{itemize}
\end{remark}

To address the question of preservation of measurement statistics, we begin
with the following lemma.

\begin{lemma}\label{Lem:measurement_statistics}
Let $\Pi$ be a projector on the Hilbert space ${\mathcal H}$.
Then
\begin{equation}
\Pi\,\phi(\rho)\,\Pi^\dag = \phi \left( \Pi\rho\Pi^\dag \right)
\label{E:projection_commutativity}
\end{equation}
for all $\rho \in {\mathcal B}({\mathcal H})$ if and only if the image of
${\mathcal H}$ under $\Pi$, that is, $\Pi[{\mathcal H}]$, is an invariant
subspace of $\phi$.
\end{lemma}

\begin{proof}
If ${\mathcal S}\equiv \Pi[{\mathcal H}]$ is an invariant subspace of $\phi$,
then part~(a) of Theorem~\ref{Thrm:structural_theorem} implies that each Kraus
operator of $\phi$ can be written as $A_i = A_i^{\mathcal S} \oplus
A_i^{{\mathcal S}^\perp}$ using the notation of that Theorem.
Regarding $\Pi$ and $A_i$'s as block matrices in a basis compatible with the
${\mathcal H} = {\mathcal S} \oplus {\mathcal S}^\perp$ direct sum structure,
\begin{align}
&\Pi\,\phi(\rho)\,\Pi^\dag \nonumber \\
= &\sum_i \begin{bmatrix} I_{\mathcal S} & \\ &
 0_{{\mathcal S}^\perp} \end{bmatrix} \begin{bmatrix} A_i^{\mathcal S} & \\ &
 A_i^{{\mathcal S}^\perp} \end{bmatrix} \rho \begin{bmatrix} A_i^{\mathcal S} &
 \\ & A_i^{{\mathcal S}^\perp} \end{bmatrix}^\dag \begin{bmatrix}
 I_{\mathcal S} & \\ & 0_{{\mathcal S}^\perp} \end{bmatrix}^\dag \nonumber \\
= &\sum_i \begin{bmatrix} A_i^{\mathcal S} & \\ & 0_{{\mathcal S}^\perp}
 \end{bmatrix} \rho \begin{bmatrix} {A_i^{\mathcal S}}^\dag & \\ &
 0_{{\mathcal S}^\perp} \end{bmatrix} \nonumber \\
= &\sum_i \begin{bmatrix} A_i^{\mathcal S} \Pi \rho \Pi^\dag
 {A_i^{\mathcal S}}^\dag & \\ & 0_{{\mathcal S}^\perp} \end{bmatrix} \nonumber
 \\
= &\sum_i \begin{bmatrix} A_i^{\mathcal S} & \\ & A_i^{{\mathcal S}^\perp}
 \end{bmatrix} \begin{bmatrix} \Pi\rho\Pi^\dag & \\ & 0_{{\mathcal S}^\perp}
 \end{bmatrix} \begin{bmatrix} A_i^{\mathcal S} & \\ & A_i^{{\mathcal S}^\perp}
 \end{bmatrix}^\dag \nonumber \\
= &\phi \left( \Pi\rho\Pi^\dag \right) \label{E:proof_commutativity_lemma1}
\end{align}
for all $\rho\in {\mathcal B}({\mathcal H})$.

To prove the converse, we observe that ${\mathcal S} = \Pi[{\mathcal H}]$ is a
Hilbert subspace of ${\mathcal H}$.
For any density matrix $\rho \in {\mathcal B}({\mathcal S})$,
Eq.~\eqref{E:projection_commutativity} becomes $\Pi\,\phi(\rho)\,\Pi^\dag =
\phi(\rho)$.
Therefore, the density matrix $\phi(\rho) \in {\mathcal B}({\mathcal S})$.
So ${\mathcal S}$ is an invariant subspace of $\phi$ by definition.
\end{proof}

\begin{theorem} \label{Thrm:measurement_statistics}
(a)~
Let $\varphi\colon {\mathcal B}({\mathcal H})\to {\mathcal B}({\mathcal H})$ be
the quantum operation $\rho \mapsto \sum_k \Pi_k \rho \Pi_k^\dag$ where
$\Pi_k$'s are projectors obeying $\sum_k \Pi_k = I_{\mathcal H}$ and $\Pi_k
\Pi_{k'} = 0$ whenever $k\neq k'$.
Then
\begin{equation}
\varphi\circ\phi(\rho) = \phi\circ\varphi(\rho) \label{E:meas_stat_channel}
\end{equation}
for all $\rho\in {\mathcal B}({\mathcal H})$ if and only if every $\Pi_k
[{\mathcal H}]$ is an invariant subspace of $\phi$.
More importantly,
\begin{equation}
\Tr \left( \Pi_k \rho \right) = \Tr \left[ \Pi_k \phi(\rho) \right]
\label{E:meas_stat_PM}
\end{equation}
for all density matrices $\rho$ in ${\mathcal B}({\mathcal H})$ and for all
$k$ if and only if $\Pi_k[{\mathcal H}]$ is an invariant subspace of $\phi$.

\noindent (b)~
More generally, consider a \POVM measurement on a quantum state in the Hilbert
space ${\mathcal H}$ with \POVM elements $\{ E_k \}$.
(Note that the number of \POVM elements need not be finite here.)
Then
\begin{equation}
\Tr \left( E_k \rho \right) = \Tr \left[ E_k \phi(\rho) \right]
\label{E:meas_stat_POVM}
\end{equation}
for all density matrices $\rho$ in ${\mathcal B}({\mathcal H})$ if and only if 
$E_k = \sum_j a_j^{(k)} P_{{\mathcal S}_j}$ with $a_j^{(k)} \geq 0$ where
$P_{{\mathcal S}_j}$ is the projector on an \IIS ${\mathcal S}_j$ of $\phi$.
\end{theorem}

\begin{proof}
To prove part~(a), note that Eq.~\eqref{E:meas_stat_channel} is a direct
consequence of Lemma~\ref{Lem:measurement_statistics} and linearity of $\phi$.

Suppose $\Pi_k[{\mathcal H}]$ is an invariant subspace of $\phi$.
By taking trace in both sides of Eq.~\eqref{E:projection_commutativity} and by
using the fact that $\phi$ is trace-preserving, we have $\Tr \left[ \Pi_k
\phi(\rho) \right] = \Tr \left[ \phi \left( \Pi_k \rho \Pi_k^\dag \right)
\right] = \Tr \left( \Pi_k \rho \Pi_k^\dag \right) = \Tr \left( \Pi_k \rho
\right)$.

To prove the converse in part~(a), suppose Eq.~\eqref{E:meas_stat_PM} holds.
We set $\rho$ to an arbitrary but fixed density matrix in ${\mathcal B}
(\Pi_k[{\mathcal H}])$.
Then, $1 = \Tr \left( \Pi_k \rho \right) = \Tr \left[ \Pi_k \phi(\rho)
\Pi_k^\dag \right]$.
Since $\phi$ is trace-preserving and positive, $\Pi_k$ is a projector and $\dim
{\mathcal H}$ is finite, we conclude that $\phi(\rho) \in {\mathcal B}
(\Pi_k[{\mathcal H}])$.
Hence, $\Pi_k[{\mathcal H}]$ is an invariant subspace of $\phi$.

We now move on to prove part~(b).
The sufficiency condition in part~(b) is a direct consequence of part~(a) which
says that $\Tr \left( P_{{\mathcal S}_j} \rho \right) = \Tr \left[
P_{{\mathcal S}_j} \phi(\rho) \right]$ for all \IIS ${\mathcal S}_j$.

To prove the condition is necessary in part~(b), we use the fact that each
$E_k$ is a positive operator in ${\mathcal B}({\mathcal H})$ and hence
self-adjoint.
As ${\mathcal H}$ is finite-dimensional, $E_k$ can be written as the finite sum
$\sum_\ell b_\ell P_{{\mathcal T}_\ell}$ where ${\mathcal T}_\ell$'s are
mutually orthogonal subspaces of ${\mathcal H}$, $b_\ell > 0$ and $b_\ell <
b_{\ell'}$ whenever $\ell > \ell'$.
(Surely, $b_\ell$'s and $T_\ell$'s depend on $k$.
But we do not explicitly emphasize this dependence to avoid clumsy notations.)
Consequently, $\Tr \left( E_k \rho \right) \leq b_1$ for all density matrices
$\rho\in {\mathcal B}({\mathcal H})$ with equality holds if and only if $\rho
\in {\mathcal B}({\mathcal T}_1)$.
Let $\rho_1 \in {\mathcal B}({\mathcal T}_1)$.
Eq.~\eqref{E:meas_stat_POVM} implies $\Tr \left[ E_k \phi(\rho_1) \right] =
b_1$.
Hence, $\phi(\rho_1)\in {\mathcal B}({\mathcal T}_1)$ and ${\mathcal T}_1$ is
an invariant subspace of $\phi$.
By Theorem~\ref{Thrm:structural_theorem}, $\phi \left[
{\mathcal B}({\mathcal T}_1^\perp) \right] \subset
{\mathcal B}({\mathcal T}_1^\perp)$.
More importantly, $\Tr \left[ E_k \rho' \right] \leq b_2$ for all density
matrices $\rho'\in {\mathcal B}({\mathcal T}_1^\perp)$ with equality holds if
and only if $\rho' \in {\mathcal B}({\mathcal T}_2)$.
So, by inductively applying the previous argument, we conclude that all
${\mathcal T}_\ell$'s are invariant subspaces of $\phi$.
This proves the converse in part~(b).
\end{proof}

\begin{remark}\label{Rem:measurement_statistics}
The above Theorem says that the order of passing a quantum state through the
unital channel $\phi$ and performing a projective measurement $\varphi$ on the
state is not important provided that $\dim {\mathcal H} < \infty$.
Moreover, let us consider the following two machines --- the first one performs
a \POVM measurement on an input quantum state and the second one performs the
same \POVM measurement after passing the input quantum state through a unital
quantum channel.
The above Theorem completely characterizes the kind of \POVM measurements in
the above two machines so that they have the same measurement statistics given
any input quantum state.
One implication of our findings is that if we do not have control on the kind
of input quantum states, then measurement statistics may be changed by a unital
quantum channel $\phi$ if the measurement is finer than the direct sum
decomposition of the underlying Hilbert space ${\mathcal H}$ into \IIS's of
$\phi$.
This implication echoes with the findings by Blume-Kohout \emph{et
al.}\cite{BK08} that the essential geometric structure underlying all
noiseless subsystems, decoherence-free subspaces, pointer bases and quantum
error-correcting codes on finite-dimensional quantum systems is an isometry to
fixed points of certain unital quantum operations.
\end{remark}

\section{Discussions}\label{Sec:discussions}

In summary, using elementary analysis and graph theoretic methods, we
completely characterize the structure of a unital quantum operation $\phi$ on
${\mathcal B}({\mathcal H})$ provided that the Hilbert space ${\mathcal H}$ is
finite-dimensional.
In particular, the basic building blocks for this kind of unital quantum
operations are those without proper \IIS induced by $\phi$ in the sense that
no convex cone formed by the set of all positive semidefinite operators acting
on a proper subspace of ${\mathcal H}$ is fixed by $\phi$.
We further show that although the direct sum decomposition of ${\mathcal H}$
into \IIS's of $\phi$ need not be unique, the number of \IIS's and the
dimension of each \IIS's are unique up to permutation of these \IIS's.
Using this structural characterization, we solve three interesting quantum
information processing problems, namely, to show a reduction theorem for
various information-theoretic capacities of a finite-dimensional unital quantum
channel, to find all the fixed states of $\phi$, and to give a necessary and
sufficient condition for the preservation of measurement statistics of a \POVM
measurement by $\phi$.

Interestingly, we find that the problem of completely characterizing $\phi$ is
reduced to the simultaneous block diagonalization by unitary conjugation of its
Kraus operators.
Note that in actual practice, this simultaneous block diagonalization can be
done relatively painlessly.
One possibility is to use a recent computationally stable algorithm reported
by Maehara and Murota in Ref.~\cite{Jap3}, which builds on their earlier works
on finding simultaneous invariant subspaces in Refs.~\cite{Jap1,Jap2}.
Another possibility is to adapt the algorithm of Blume-Kohout \emph{et al.}
in Refs.~\cite{BK08,BK10} on finding the so-called information-preserving
structures of a quantum channel.
Thinking along this direction, the following result of Mendl and
Wolf~\cite{MendlWolf} may be of interest.
They showed that a unital quantum operation can always be expressed as the sum
$\sum_i a_i U_i \cdot U_i^\dag$ where $a_i\in {\mathbb R}$ obeying $\sum_i a_i
= 1$.~\cite{MendlWolf}
Note that this is not an operator-sum representation for $a_i$ can be negative.
It is instructive to see how to efficiently convert an operation-sum
representation of a unital quantum operation $\phi$ to the above affine sum of
unitary conjugations in a computationally stable manner.
It is also instructive to see if Kribs'~\cite{Kribs03} and
Arias \emph{et al.}'s~\cite{AriasJMP02} necessary and sufficient condition on
fixed points of a unital quantum operation restated as Theorem~\ref{Thrm:Arias}
here can be modified to cover the case of the affine sum representation of
Mendl and Wolf.
Since finding simultaneous diagonal blocks of a set of unitary matrices is much
simpler problem, solving the above two problems may improve our computational
efficiency and accuracy in finding \IIS's of $\phi$.

\begin{acknowledgments}
We would like to thank Debbie W.\ C.\ Leung for her valuable discussions.
This work is supported under the RGC grant HKU~700709P of the HKSAR Government.
\end{acknowledgments}

\bibliography{qc54.6}
\end{document}